\documentclass[11pt]{article}

\usepackage{amsmath}
\usepackage[mathscr]{eucal}
\usepackage{amssymb}
\usepackage{graphicx}
\usepackage{latexsym}
\usepackage{amsthm}
\usepackage{geometry}

\usepackage{tikz}
\usetikzlibrary[shapes.geometric]
\usepackage{tikz-cat}

\numberwithin{equation}{section}
\newtheorem{theorem}{Theorem}[section]
\newtheorem{lemma}{Lemma}[section]
\newtheorem{corollary}{Corollary}[section]
\newtheorem{definition}{Definition}[section]
\title{Which graph states are useful for quantum~information~processing?}
\author{Mehdi~Mhalla\footnote {CNRS, LIG, Universit\'e de Grenoble, France},~Mio~Murao\footnote {Graduate School of Science, The University of Tokyo, Japan} \footnote{NanoQuine, The University of Tokyo, Japan}, \\Simon~Perdrix\footnotemark[1],~Masato~Someya\footnotemark[2],~and~Peter~S.~Turner\footnotemark[2]}
 
\date{}
\begin{document}
\maketitle


\newcommand{\kl}{\left|}
\newcommand{\smat}[3]{{#1}|_{#3}^{#2}}
\newcommand{\ket}[1]{\left| #1 \right\rangle}
\newcommand{\bra}[1]{\left\langle #1 \right|}
\newcommand{\kr}{\right\rangle}
\newcommand{\bl}{\left\langle}
\newcommand{\bral}{\langle}
\newcommand{\br}{\right|}
\newcommand{\bs}{\{ 0,1 \}}
\newcommand{\bplus}{\left\langle +\cdot \cdot \cdot + \right|}
\newcommand{\kplus}{\left| +\cdot \cdot \cdot + \right\rangle}
\newcommand{\comm}[1]{ {\bf \hspace{-3cm}#1}}

\begin{abstract}
Graph states~\cite{HEB04} are an elegant and powerful quantum resource for measurement based quantum computation (MBQC). They are also used for many quantum protocols (error correction, secret sharing, etc.). 
The main focus of this paper is to provide a structural characterisation of the graph states that can be used for quantum information processing. 
The existence of a gflow (generalized flow)  \cite{MP08} is known to be a requirement for open graphs (graph, input set and output set) to perform  uniformly and strongly  deterministic computations. We weaken the gflow conditions  to define two new more general kinds of MBQC: uniform equiprobability and cons\-tant probability.  
These classes can be useful from a cryptographic and information point of view because even though we cannot do a deterministic computation in general we can preserve the  information and transfer it perfectly from the inputs to the outputs.
We derive  simple graph characterisations for these classes and  prove that the deterministic and uniform equiprobability classes collapse when the cardinalities of inputs and outputs are the same. We also prove the reversibility of gflow in that case.
The new graphical characterisations  allow us to go from open graphs to  graphs in general and to consider this question: given a graph with no inputs or outputs fixed, which vertices can be chosen as input and output for quantum information processing? We present a characterisation of the sets of possible inputs and ouputs for the equiprobability class, which is also valid for deterministic computations with  inputs and ouputs of the same cardinality. 
\end{abstract}

\section{Introduction}

The graph state formalism~\cite{HEB04} is an elegant and powerful formalism for quantum information processing. Graph states form a subfamily of the stabiliser states \cite{got}. They provide a graphical description of entangled states and they have multiple applications in quantum information processing, in particular in measurement-based quantum computation (MBQC) \cite{RB01}, but also in quantum error correcting codes \cite{got} and in quantum protocols like secret sharing \cite{MS08,KMMP}. They offer a combinatorial approach to the characterisation of the  fundamental properties of entangled states in quantum information processing. The invariance of the entanglement by local complementation of a graph \cite{VdN}; the use of measure of entanglement based on the rank-width of a graph \cite{VdN2};  and the combinatorial \emph{flow} characterisation \cite{BKMP} of deterministic evolutions in measurement-based quantum computation witness the import role of the graph state formalism in quantum information processing. 

In this paper, we focus on the application of graph states in MBQC and in particular on the characterisation of graphs that can be used to perform quantum information processing in this context. The existence of a graphical condition which guarantees that a deterministic MBQC evolution can be driven despite of the probabilistic behaviour of the measurements is a central point in MBQC. It has already been proven that the existence of a certain kind of flow called glfow characterises uniformly stepwise determinism \cite{BKMP}. In section \ref{sec:det}, we introduce a simpler but equivalent combinatorial characterisation using \emph{focused gflow} and we provide a simple condition of existence of such a flow as the existence of a right inverse to the adjacency matrix of the graph. We also prove additional properties in the case where the number of input and output qubits of the computation are the same: the gflow is then reversible and the stepwise condition \cite{BKMP} on determinism is not required to guarantee the existence of a gflow.

  The main contribution of this paper is the weakening of the determinism condition in order to consider the more general class of \emph{information preserving} evolutions. Being information preserving is one of the most fundamental property that can be required for a MBQC computation. 
  Indeed, some non-deterministic evolutions can be information preserving when one knows the classical outcomes of the measurements produced by the computation. Such evolutions are called \emph{equi-probabilistic} -- when each classical outcome  occurs with probability $1/2$ -- or \emph{constant-probabilistic} in the general case. 
  In section \ref{sec:relax}, we introduce simple combinatorial conditions for {equi-proba\-bi\-lis\-tic} and  {constant-probabilistic} MBQC by means of excluded violating sets of vertices. We show, in the particular case where the number of input and output qubits are the same, that graphs guaranteeing equi-probabilism and determinism are the same. 
  In section \ref{sec:choosing}, using this graphical characterisation, we address the fundamental question of finding input and output vertices in an arbitrary graph for guaranteeing an equi-probabilistic (or deterministic) evolution. To this end, we show that the input and output vertices of a graph must form  transversals of the violating sets induced by  the equi-probabilistic characterisation. 
  Finally, in the last section, we  investigate several properties of  the most general and less understood class of constant probabilistic evolutions.

\section{Measurement-based quantum computation}\label{sec:mbqc}

In this section, the main ingredients of measurement based quantum computation (MBQC) are described. More detailed introductions can be found in \cite{DKP,DKPP}.  An MBQC is described by $(i)$ an open graph $(G,I,O)$ ($G$ is a simple undirected graph, $I,O\subseteq V(G)$ are called resp. input and output vertices); $(ii)$ a map $\alpha :  O^C\to [0,2\pi)$, where $O^C := V\setminus O$, which associates with every non ouput vertex an angle; and $(iii)$ two maps $\mathtt x, \mathtt z: O^C \to \{0,1\}^{V(G)}$ called \emph{corrective maps}. A vertex $v\in supp(\mathtt x(u))\cup supp(\mathtt z(u))$ is called a \emph{corrector} of $u$, where $supp(y) = \{u ~|~ y_u =1\}$.  The maps $\mathtt x, \mathtt z$ should be \emph{extensive} in the sense that there exists a partial order $\prec$ over the vertices of the graph s.t. any corrector $v$ of a vertex $u$ is larger than $u$, i.e. $v\in supp(\mathtt x(u))\cup supp(\mathtt  z(u))$ implies  $u\prec v$.

Let $N : \mathbb C^{\{0,1\}^I}\to \mathbb C^{\{0,1\}^{V(G)}}$ be the preparation map which associates with any arbitrary input state located on the input qubits the initial entangled state of the MBQC:  $$N=  \frac{1}{\sqrt{2^{|I^C|}}} \sum_{x\in \{0,1\}^I, y\in \{0,1\}^{I^C}} (-1)^{q(x,y)}\ket {x,y}\bra x$$

 where $q: \{0,1\}^{V(G)}\to \mathbb N = x\mapsto  |E(G)\cap (supp(x)\times supp(x))|$ associates with every $x$  the number of edges of the subgraph $G_x=(V(G)\cap supp(x), E(G) \cap (supp(x)\times supp(x)))$ induced by $x$.
 
The one-qubit measurements, parametrized by an angle $\alpha_u$, of every non-output qubit $u$ are inducing the following projection $P_s(\alpha) : \mathbb C^{\{0,1\}^{V(G)}}\to  \mathbb C^{\{0,1\}^{O}}$ of the entangled state onto the subspace of the output qubits, where  $s\in \{0,1\}^{O^C}$ stands for the classical outcomes of the one-qubit measurements:
\begin{eqnarray*}
P_s(\alpha) &=& \frac{1}{\sqrt {2^{|I^C|}}}\sum_{x\in \{0,1\}O^C,y\in \{0,1\}^O} e^{\alpha_{x\cdot s}}\ket{y}\bra {xy}\end{eqnarray*}
with 
$\alpha_x = \sum_{u\in supp(x)}\alpha(u)$ and $x\cdot s$ is the bitwise conjonction of $x$ and $s$.


Moreover, adaptative Pauli corrections depending on the classical outcomes of the measurements and on the corrective maps, are applied during the computation leading, for any possible classical outcomes $s\in \{0,1\}^{O^C}$, to the following overall (postselected) evolution $\chi_s : \mathbb C^{\{0,1\}^{I}}\to  \mathbb C^{\{0,1\}^{O}}$:
$$\chi_s=P_s(\alpha)\left(\prod_{u\in V(G)}X_{s\cdot \mathtt x(u)}Z_{s\cdot \mathtt z(u)}\right)N$$
where $X_s$ and $Z_s$ are Pauli operators: $X_s=\bigotimes_{u\in supp(s)} X_u$ and $Z_s=\bigotimes_{u\in supp(s)} Z_u$. 

An MBQC is implementing the quantum operation $\{\chi_s\}_{s\in \{0,1\}^{O^C}}$. The evolution is as follows: a classical outcome (also called branch) $s\in \{0,1\}^{O^C}$ is produced and the input state $\ket \phi \in \mathbb C^{\{0,1\}^I}$ is mapped to the state $\chi_s\ket \phi\in \mathbb C^{\{0,1\}^O}$ (up to a normalisation). The probability for an outcome  $s\in \{0,1\}^{O^C}$ to occur is $p_s = ||\chi_s\ket \phi ||^2$.

The overall evolution can be decomposed into several steps, corresponding to a possible implementation of the MBQC model:  first the input state $\ket \phi$ is encoded into the open graph state $\ket {\phi_G} = N\ket \phi$, then the  local measurements (qubit $u$ is measured according the observable $\cos(\alpha(u))X + \sin(\alpha(u))Y$) and the  local Pauli corrections  are performed. This sequence of local operations is done according to the partial order induced by the correction maps  $\mathtt x, \mathtt z$. 

\section{Determinism}\label{sec:det}

\begin{definition}
An MBQC $(G,I,O,\alpha, \mathtt x, \mathtt z)$ is strongly deterministic if all the branches are implementing the same map, i.e. $\exists U$ s.t.  $\forall s\in \{0,1\}^{O^C}$, $\chi_s = \frac1{\sqrt{2^{|O^C|}}}U$. 
\end{definition}

\begin{lemma} 
If an MBQC is strongly deterministic then it implements an isometry.
\end{lemma}

\begin{proof}
Since $\sum_{s\in \{0,1\}^{O^C}} \chi_s^\dagger  \chi_s =I$, $U^\dagger U=I$ so $U$ is an isometry and the MBQC implements the super operator $\rho \mapsto U\rho U^\dagger$.
\end{proof}

In order to point out the combinatorial properties of MBQC, the angles of measurements and the corrective maps can be abstracted away in the following way in order to keep only the influence of the  initial open graph.

\begin{definition}
An open graph $(G,I,O)$ guarantees uniformly strong determinism if $\exists \mathtt x, \mathtt z$ s.t. $\forall \alpha$, $(G,I,O,\alpha, \mathtt x, \mathtt z)$ is strongly deterministic.
\end{definition}

An MBQC is said to guarantee \emph{stepwise} strong determinism if any partial computation is also strongly deterministic.

The gflow of an open graph is defined as follows, based on the use of the odd neighborhood of a set of vertices: for a given subset $S$ of vertices in a graph $G$, $Odd(S):=\{v\in V(G) ~s.t.~ |N(v)\cap S|=1~[2]\}$.

\begin{definition}
$(g,\prec)$ is a  \emph{gflow} of $(G,I,O)$, where $g:  O^C \to 2^{I^c}$, if for any $u$, 
\\~~~~ --- if $ v \in g(u)$, then $u \prec v$;
\\~~~~ --- $u\in  Odd(g(u))$; 
\\~~~~ --- if $v\in  Odd(g(u))$ and $u\neq v$ then $u\prec v$.
\end{definition}

\begin{theorem}[\cite{BKMP}]\label{thm:gflow}
An open graph $(G,I,O)$ guarantees uniform stepwise  strong determinism iff $(G,I,O)$ has a gflow.
\end{theorem}
\subsection{Focused gflow}

Since the gflow is not unique we introduce a stronger version called focused gflow, which is unique if the number of inputs and outputs are the same. The focused gflow gives rise to a simpler characterisation of uniform stepwise  strong determinism. The focused gflow is based on the use of \emph{extensive} maps.
\begin{definition} $g: O^C \to 2^{I^C}$ is a focused gflow of $(G,I,O)$ if $g$ is extensive -- i.e. the transitive closure of the relation $\{(u,v) ~s.t.~ v\in g(u)\}$ is a partial order over $V(G)$ -- and $\forall u\in O^C$, $Odd(g(u))\cap O^C = \{u\}$
\end{definition}

\begin{theorem}\label{thm:focused}
An open graph $(G,I,O)$ guarantees uniform stepwise  strong determinism iff $(G,I,O)$ has a focused gflow.
\end{theorem}

\begin{proof} We prove that $(G,I,O)$ has a gflow iff it has a focused gflow. First, assume $g$ is a focused gflow, and let $\prec$ be the transitive closure of $\{(u,v) ~s.t.~ v\in g(u)\}$. $\prec$ is a partial order and by definition, if $v\in g(u)$ then $u\prec v$. Moreover $u\in Odd(g(u))=\{u\}$. Finally, if $v \in Odd(g(u))$ and $v\neq u$ then $v\in O$, so there is no element larger than $v$ by definition of $\prec$. Thus $(g,\prec)$ is a gflow.  
Now, assume $(g,\prec)$ is a gflow. We call the depth of a vertex  $u$ its distance to the output, i.e. the length of longest strictly increasing sequence $u\prec u_1\prec..\prec u_k$ s.t. $u_k\in O$. We construct a focus gflow $g_f$ by induction on the depth of the vertices. If $u$ is of depth $1$ then $g_f(u):=g(u)$. If $u$ is of depth larger than $2$, let $g_f(u):=g(u)\bigoplus_{v\in Odd(g(u))\cap O^C, v\neq u} g_f(v)$, where $\oplus$ is the symmetric difference: $A\oplus B = (A\cup B)\setminus (A\cap B)$.  Since $Odd(A\oplus B)=Odd(A)\oplus Odd(B)$, $Odd(g_f(u))\cap O^C = (Odd(g(u)) \bigoplus_{v\in Odd(g(u))\cap O^C, v\neq u} Odd(g_f(v)))\cap O^C=(Odd(g(u))\cap O^C)\oplus (Odd(g(u))\setminus \{u\})\cap O^C) = \{u\}$. Moreover $g_f$ is extensive since the relation $R$ induced by $g_f$ is s.t. $u R vÊ\implies u\prec v$ so the transitive closure of $R$ is a partial order. 
\end{proof}

\subsection{Induced adjacency matrix and reversibility}

We introduce the notion of induced adjacency matrix of an open graph and show that an open graph has a gflow if and only if its induced matrix has a DAG as right inverse.

\begin{definition}
The induced adjacency matrix of an open graph $(G,I,O)$ is the submatrix $\smat{A_{G}}{O^C}{I^C}$
 of the adjacency matrix $A_G=\{m_{u,v}, (u,v)\in V(G)\}$ of $G$ removing the rows of $O$ and column of $I$, i.e. $\smat{A_{G}}{O^C}{I^C} = \{m_{u,v}, (u,v)\in O^C\times I^C\}$. 
\end{definition}

The induced matrix $\smat{A_{G}}{O^C}{I^C}$ is the matrix representation of the linear map $W\mapsto Odd(W)\cap I^C$ which domain is $2^{O^C}$ and codomain is $2^{I^C}$.

\begin{theorem}\label{thm:rightinv}
$(G,I,O)$ has a gflow iff there exists a DAG\footnote{DAG: Directed Acyclic Graph} $F=(V(G),E)$ s.t. $$\smat{A_G}{O^C}{I^C}. \smat{A_F}{I^C}{O^C} = I$$
\end{theorem}

\begin{proof}
(\emph{only if}) Assume $(G,I,O)$ has a gflow. Thanks to lemma \ref{thm:focused} w.l.o.g. $(G,I,O)$ has a focused gflow $g_f$. Let $F=(V(G),E)$ be a directed graph s.t. $(u,v)\in E(F) \iff v\in g_f(u)$. 
Notice that  $\forall u\in O^C$, $\smat{A_F}{I^C}{O^C}1_{\{u\}} = 1_{g_f(u)}$ where $1_X$ is a binary vector s.t. $({1_X})_u =1 \iff u\in X$. Moreover, since $g_f$ is extensive, $F$ is a DAG. Thus $\smat{A_G}{O^C}{I^C}\smat{A_F}{I^C}{O^C} 1_{\{u\}} = \smat{A_G}{O^C}{I^C}1_{g(u)}= 1_{Odd(g_f(u))\cap O^C}=1_{\{u\}}$. (\emph{if}) Assume $F=(V(G),E)$ be a DAG s.t. $\smat{A_G}{O^C}{I^C}. \smat{A_F}{I^C}{O^C} = I$, then let $g:O^C\to 2^{I^C}= u\mapsto N_F^+(u)$. Since $F$ is a DAG, $g$ is extensive, and $1_{Odd(g(u))\cap O^C}=\smat{A_F}{I^C}{O^C}(1_{g(u)})= \smat{A_G}{O^C}{I^C}\smat{A_F}{I^C}{O^C} 1_{\{u\}} = 1_{\{u\}}$, so $Odd(g(u))\cap O^C=\{u\}$. 
\end{proof}

Thus, according to theorem \ref{thm:rightinv}, an open graph has a gflow if and only if it has a DAG as right inverse. Notice that this DAG is nothing but the graphical description of the focused gflow function: the set of successors of a vertex $u$ is the image of $u$ by the focused gflow function.


As a corollary of Theorem \ref{thm:rightinv}, $(G,I,O)$ has no gflow if $|I|>|O|$. Indeed, for dimension reasons, if $|I|>|O|$ the matrix $\smat{A_G}{O^C}{I^C}$ has no right inverse. When $|I|=|O|$ the focused gflow is \emph{reversible} in the following sense: 

\begin{theorem} \label{thm:revers}
When $|I|=|O|$, $(G,I,O)$ has a gflow iff $(G,O,I)$ has a gflow.
\end{theorem}

\begin{proof} Notice that the induced adjacency matrix of $(G,O,I)$ is the transpose $\,^t \smat{A_G}{O^C}{I^C}$ of the one of $(G,I,O)$. Moreover, since $\smat{A_G}{O^C}{I^C}$ is squared, $\smat{A_F}{I^C}{O^C}$ is both right and left inverse of $\smat{A_G}{O^C}{I^C}$.  
Thus, $\smat{A_G}{I^C}{O^C}. ^t\smat{A_F}{I^C}{O^C} = \,^t(\smat{A_F}{I^C}{O^C}.\smat{A_G}{O^C}{I^C})= I$. As a consequence $\smat{A_G}{O^C}{I^C}$ has a right inverse which is a DAG since the transpose of a DAG is a DAG. 
\end{proof}


\section{Relaxing Uniform Determinism}\label{sec:relax}
Focused gflow guarantees uniformly stepwise strong determinism. 
We consider here two more general classes of MBQC evolutions: the \emph{equi-probabilistic} case where all the branches occur with the same probability, independent of the input state; and the \emph{constant probability} case where all the branches occur with a probability independent of the input state. We show that both equi-probabilitic and constant probabilistic evolutions are information preserving and admit a simple graphical characterisation by means of violating sets.

\begin{definition}
An MBQC $(G,I,O,\alpha, \mathtt x, \mathtt z)$ is:\\
--- \emph{equi-probabilistic} if for any input state $\ket \phi \in \mathbb C^{2^I}$ and any branch $s\in \{0,1\}^{O^C}$, $p_s= ||\chi_s\ket \phi|| = \frac 1{2^{|O^C|}}$.\\
--- \emph{constant-probabilistic} if for any branch $s\in \{0,1\}^{O^C}$ the probability $p_s=||\chi_s\ket \phi|| $ that the branch $s$ occurs does not depend on the input state $\ket \phi$. 
\end{definition}

Constant probabilistic (and hence equi-probabilistic) evolutions are \emph{information preserving} in the sense that if one knows the branch $s$ of the computation (i.e. the classical outcome) then he can recover the initial input state of the computation. Indeed, if an MBQC is constant probabilistic then the map $\ket \phi \mapsto  ||\chi_s\ket \phi||$ is constant, thus $\chi_s^\dagger \chi_s = p_s. I$. If $p_s=0$ then the branch never occurs, otherwise the branch $s$ is implementing an isometry. 

~

\noindent \emph{Remark:} Notice that the knowledge of the branch $s$, which is necessary the case in the MBQC model because of the corrective strategy, is essential to make an  equi-probabilistic evolution information preserving. Indeed, consider the quantum one-time pad example with $\forall s\in \{0,1\}^2$, $\chi_s = \sigma_s/2$ where $\sigma_s$ is a Pauli operator ($\sigma_{00}=I$,   $\sigma_{01}=X$,$\sigma_{10}=Y$,$\sigma_{11}=Z$). This evolution is equi-probabilistic but if the information of the branch is not taken into account, the corresponding super operator is $\rho \mapsto \sum_{s\in \{0,1\}^2} \sigma_s \rho\sigma_s^\dagger = I/2$ which is clearly not information preserving. 

~

We prove that uniform equi- and constant probabilities
 have simple graph characterisations by violating sets, where uniformity is defined similarly to the determinism case: $(G,I,O)$ guarantees uniformly constant (resp. equi-) probabilisty if  $\exists \mathtt x, \mathtt z$ s.t. $\forall \alpha$, $(G,I,O,\alpha, \mathtt x, \mathtt z)$ has a constant (resp. equi-) probabilistic evolution.
  
\begin{theorem} \label{uu}
An open graph $(G,I,O)$ guarantees uniform equiprobability iff
$$\forall W\subseteq O^C, Odd(W)\subseteq W\cup I \implies W =\emptyset$$
\end{theorem}

A nonempty set $W\subseteq O^C$ such that $Odd(W)\subseteq W\cup I $ is called an internal set. Theorem \ref{uu} says that an open graph $(G,I,O)$ guarantees uniform equi-probability if and only if it has no internal set.  

\begin{proof}
(\emph{if}) First we assume that there is no internal set  and we show that every branch occurs with the same probability $1/{2^{|O^C|}}$, independently of the input state and the set of measurement angles. For a given open graph $(G,I,O)$, a given input state $\ket \phi$ and a given set of measurement angles $\{\alpha_v\}_{v\in O^C}$, we consider w.l.o.g. the $0$-branch, i.e. the branch where all outcomes are $0$ \footnote{The other branches are taken into account by considering a different set of measurement angles e.g. the branch where all outcomes are $1$ corresponds to the $0$-branch when the set of measurements is $\{\alpha_v +\pi\}_{v\in O^C}$.}. The probability of this branch is $p = || \prod_{v\in O^c} {\bral +_{\alpha_v}}\ket {\phi_{G}}||^2 = \frac1{{2^{|O^C|}}}||\sum_{x\in \{0,1\}^{O^C}} e^{i\alpha_x}\bral x \ket {\phi_G}||^2$ where $\alpha_x = \sum_{v\in O^C} \alpha_v.x_v$ and $\ket{\phi_G} = E_G\ket +_{I^C}\ket \phi_I$. As a consequence,\\ 
$\begin{array}{rcl} p &=& \frac1{{2^{|O^C|}}}\sum_{x,y\in \{0,1\}^{O^C}} e^{i(\alpha_y-\alpha_x)}\bral {\phi_G}\ket x \bral y \ket {\phi_G}\\
&=&\frac1{{2^{|O^C|}}} \sum_{u\in \{-1,0,1\}^{O^C}} e^{i\alpha_u}\sum_{x,y\in \{0,1\}^{O^C}~s.t.~x-y=u} \bral {\phi_G}\ket x \bral y \ket {\phi_G} \\
&=& \frac1{{2^{|O^C|}}} \sum_{u\in \{-1,0,1\}^{O^C}} e^{i\alpha_u}\sum_{x\in \{0,1\}^{V_u^C}} \bral {\phi_G}\ket x_{V_u^C} \ket {\frac {1+u}2}_{V_u} \bra x_{V_u^C} \bra {\frac {1-u}2}_{V_u}\ket {\phi_G} \\
&=&\frac1{{2^{|O^C|}}}\sum_{u\in \{-1,0,1\}^{O^C}} e^{i\alpha_u}\bra {\phi_G} \ket {\frac {1+u}2}_{V_u}\left(\sum_{x\in \{0,1\}^{V_u^C}} \ket x \bra x \right) \bra {\frac {1-u}2}_{V_u}\ket {\phi_G}\\
&=&\frac1{{2^{|O^C|}}}\sum_{u\in \{-1,0,1\}^{O^C}} e^{i\alpha_u}\bra {\phi_G} \ket {\frac {1+u}2}_{V_u} \bra {\frac {1-u}2}_{V_u}\ket {\phi_G}\\
&=&\frac1{{2^{|O^C|}}}\sum_{u\in \{-1,0,1\}^{O^C}} e^{i\alpha_u}p_u
\end{array}$\\
where $V_{u} =\{i\in O^C ~|~ u_i\neq 0\}$,  $\ket{\frac {1\pm u}2}_{V_u} =\bigotimes_{i\in V_u} \ket {\frac {1\pm u_i}2}_i$, and $p_u = \bra {\phi_G} \ket {\frac {1+u}2}_{V_u} \bra {\frac {1-u}2}_{V_u}\ket {\phi_G}$. 
Notice that for any $v\in I^C$, $\ket {\phi_G} = \frac1{\sqrt 2}\sum_{a\in \{0,1\}}Z^a_{N_G(v)}\ket {\phi_{G\setminus v}}\otimes \ket a_v$. 
Thus for any $u\in \{-1,0,1\}^{O^C}$ s.t. $V_u\neq \emptyset$, 
there exists $v\in  I^C\cap V_{u}^C \cap Odd(V_{u})$ (which is not empty by hypothesis) such that:\\
$\begin{array}{rcl}
p_u&=& \bra {\phi_G} \ket {\frac {1+u}2}_{V_u} \bra {\frac {1+u}2}_{V_u}X_{V_u} \ket {\phi_G}\\
&=&\frac 12 \sum_{a,b\in \{0,1\}} \bra {\phi_{G\setminus v}} \bra a_vZ^a_{N_G(v)}  \ket {\frac {1+u}2}_{V_u} \bra {\frac {1+u}2}_{V_u}X_{V_u}Z^b_{N_G(v)}\ket {\phi_{G\setminus v}}\ket b_v \\
&=&\frac 12 \sum_{a\in \{0,1\}} \bra {\phi_{G\setminus v}} Z^a_{N_G(v)}  \ket {\frac {1+u}2}_{V_u} \bra {\frac {1+u}2}_{V_u}X_{V_u}Z^a_{N_G(v)}\ket {\phi_{G\setminus v}} \\
&=&\frac 12 \sum_{a\in \{0,1\}} (-1)^a \bra {\phi_{G\setminus v}} Z^a_{N_G(v)}  \ket {\frac {1+u}2}_{V_u} \bra {\frac {1+u}2}_{V_u}Z^a_{N_G(v)}X_{V_u}\ket {\phi_{G\setminus v}} \\
&=&\frac 12 \sum_{a\in \{0,1\}}(-1)^a \bra {\phi_{G\setminus v}}   \ket {\frac {1+u}2}_{V_u} \bra {\frac {1+u}2}_{V_u}X_{V_u}\ket {\phi_{G\setminus v}} = 0
\end{array}$ \\
where the factor $(-1)^a$ comes  from the fact that $X_{V_u}$ and $Z^a_{N_G(v)}$ are commuting when $a=0$ and anticommuting when $a=1$ since $v\in Odd(V_{u})$. 
  As a consequence, it remains in $p$ only the case where $V_u=\emptyset$, so $p=  
   \frac1{{2^{|O^C|}}}\bral {\phi_G} \ket {\phi_G} = \frac1{{2^{|O^C|}}}$.
   
(\emph{only if}) Now we prove that the existence of a violating set  implies that there exists a particular input state  and a particular set of measurement angles such that some branches occur with probability $0$. Let $W_0 \subseteq O^C$ s.t. $Odd(W_0)\cap W_0^C \cap I^C = \emptyset$ and $P=\bigotimes_{v\in V} P_v$ be a Pauli operator defined as follows:\\ \centerline{$\forall v\in V, P_v=\begin{cases} X&\text{ if $v\in W_0$ and $v\notin  Odd(W_0)$}\\ Y&\text{ if $v\in W_0\cap Odd(W_0) $}\\ I& \text{ otherwise}\end{cases}$}
Let $\ket {\phi_0} = \ket {+}_{W_0\cap I}\otimes \ket 0_{W_0^C\cap I}$ be an input state. Notice that\\
$\begin{array}{rcl}
PE_G\ket+_{I^C}\ket {\phi_0} &=& (-1)^{|E(W_0)|}E_GX_{W_0}Z_{Odd(W_0)\cap W_0^C}\ket+_{I^C}\ket {\phi_0} \\
&=&(-1)^{|E(W_0)|}E_GX_{W_0}\ket+_{I^C\cup W_0}Z_{Odd(W_0)\cap W_0^C}\ket 0_{W_0^C\cap I} \\
&=& (-1)^{|E(W_0)|}E_G\ket+_{I^C}\ket {\phi_0},
\end{array}$ \\
where $E(W) = E\cap (W\times W)$ is the set of the internal edges of $W$.  Thus $E_G\ket+_{I^C}\ket \phi_0$ is an the eigenvector of $P$ associated with the eigenvalue $(-1)^{|E(W_0)|}$, implying that if each qubit $v\in W_0$ is individually measured according to the observable $P_v$ producing the classical outcome $s_v\in \{0,1\}$,  then $\sum_{v\in W_0} s_v =|E(W_0)|~[2]$. 
As a consequence, for the input $\ket {\phi_0}$ and any set of measurements $\{\alpha_v\}_{v\in O^C}$ s.t. $\alpha_v = 0$ if $v\in W_0\cap Odd(W_0)^C $ and $\alpha_v = \pi/2$ if $v\in W_0\cap Odd(W_0)$,  all the  branches ${\bf s}$ s.t. $\sum_{v\in W_0} {\bf s}_v = 1+|E(W_0)|~ [2]$ occur with probability $0$.
\end{proof}

\begin{theorem}\label{vv}
An open graph $(G,I,O)$ guarantees uniform constant probability iff 
$$\forall W\subseteq O^C, Odd(W)\subseteq W\cup I \implies L(W) \cap I = \emptyset$$
where $L(W) := \mbox{Odd}(W) \cup W$ denotes a local set. 
\end{theorem}

A nonempty set $W\subseteq O^C$ such that $Odd(W)\subseteq W\cup I $ and $L(W)\cap I \neq \emptyset$ is called a strongly internal set. Theorem \ref{vv} says that an open graph $(G,I,O)$ guarantees uniform constant probability if and only if it has no strongly internal set, or equivalently if and only if all internal sets are `far enough' from the inputs.

\begin{proof} 
(\emph{if}) First we assume that there is no strongly internal set and we show that every branch occurs with a probability independent of the input. Using the notations of the proof of theorem \ref{uu}, it only  remains to prove that $p_u$ is independent of the input for any $u\neq 0$ such that $I^C\cap V_u^C\cap Odd(V_u)=\emptyset$ and $L(V_u)\cap I=\emptyset$. Note that $Odd(V_u)\subseteq V_u\subseteq I^C$ so \\
$\begin{array}{rcl}
p_u&=& \bra {\phi_G} \ket {\frac {1+u}2}_{V_u} \bra {\frac {1+u}2}_{V_u}X_{V_u} \ket {\phi_G}\\
&=& (-1)^{|E(V_u)|}\bra {\phi_G} \ket {\frac {1+u}2}_{V_u} \bra {\frac {1+u}2}_{V_u}E_GZ_{Odd(V_u)} X_{V_u}\ket +_{I^C}\ket {\phi}_I\\
&=& (-1)^{|E(V_u)|}\bra {\phi_G} \ket {\frac {1+u}2}_{V_u} \bra {\frac {1+u}2}_{V_u}E_GZ_{Odd(V_u)} \ket +_{I^C}\ket {\phi}_I\\
&=& (-1)^{|E(V_u)|+|V_u\cap Odd(V_u)|}\bra {\phi_G} \ket {\frac {1+u}2}_{V_u} \bra {\frac {1+u}2}_{V_u}\ket {\phi_G}\\
\end{array}$ \\
Moreover,   for any $v\in V_u$, since $v\in I^C$, $\bra {\phi_G} \ket {\frac {1+u}2}_{V_u} \bra {\frac {1+u}2}_{V_u}\ket {\phi_G}$\\$\begin{array}{rcl}&=&\frac1{2}\sum_{a,b\in \{0,1\}} \bra {\phi_{G\setminus v}}\bra a_v Z^a_{N_G(v)}\ket {\frac {1+u}2}_{V_u} \bra {\frac {1+u}2}_{V_u}Z^b_{N_G(v)} \ket b_v \ket {\phi_{G\setminus v}}\\
&=&\frac1{2} \bra {\phi_{G\setminus v}} Z^{\frac {1+u_v}2}_{N_G(v)}\ket {\frac {1+u}2}_{V_u\setminus v} \bra {\frac {1+u}2}_{V_u\setminus v}Z^{\frac {1+u_v}2}_{N_G(v)}  \ket {\phi_{G\setminus v}}\\
&=& \frac1{2}\bra {\phi_{G\setminus v}} \ket {\frac {1+u}2}_{V_u\setminus v} \bra {\frac {1+u}2}_{V_u\setminus v}  \ket {\phi_{G\setminus v}}\\
\end{array}$\\
So, by induction, $\bra {\phi_G} \ket {\frac {1+u}2}_{V_u} \bra {\frac {1+u}2}_{V_u}\ket {\phi_G} =\frac1{2^{|V_u|}} \bra {\phi_{G\setminus V_u}}\ket {\phi_{G\setminus V_u}} = \frac1{2^{|V_u|}}$. 
This shows that $p_u$ does not depend on the input state. \\
(\emph{only if}) Now we prove that the existence of a strongly internal set implies that there exists a particular set of measurement angles such that some branches occur with probability zero for some input state and with nonzero probability for other inputs. Let $W_0 \subseteq O^C$ s.t. $Odd(W_0)\cap W_0^C \cap I^C = \emptyset$,  $u_0\in L(W_0)\cap I$, and  $P=\bigotimes_{v\in V} P_v$ be a Pauli operator defined like in the proof of theorem \ref{uu}. 
We consider the following input states:
$\ket {\phi_a} = \ket {+}_{W_0\cap I}\otimes \ket 0_{W_0^C\cap I\setminus u_0}\otimes \ket a_u$ for $a\in \{0,1\}$. Notice that $PE_G\ket+_{I^C}\ket {\phi_a}_I =(-1)^{a+|E(W_0)|}  E_G\ket+_{I^C}\ket {\phi_a}_I$. 
Let $\alpha_v = \pi/2$ if $v\in W_0\cap Odd(W_0)$ and $\alpha_v = 0$ otherwise. We consider a branch ${\bf s}$ of measurement which occurs with a nonzero probability if the input state  is $\ket {\phi_0}$. Notice that this branch satisfies $\sum_{v\in W_0} {\bf s}_v = (-1)^{|E(W_0)|}$. As a consequence, if the input state is $\ket {\phi_1}$, this branch ${\bf s}$ occurs with probability $0$.
\end{proof}

\section{Uniform Equiprobability versus Gflow Existence}

Since the existence of a gflow implies strongly uniform determinism it also implies uniform equiprobability. In general uniform equiprobability does not imply gflow: 

\begin{lemma}
When $|I|\neq |O|$, there exists an open graph that satisfies uniform equiprobability but that has no gflow. 
\end{lemma}
\begin{proof}
Consider the graph depicted  in Figure \ref{equigf}.
It is easy to see that it has no gflow, as no subset of the outputs has a single vertex as its odd neighorhood. On the other hand, all the  subsets of $O^C$ have a nonempty external odd neighborhood in $I^C$.
\end{proof}

\begin{figure}[!h] 
\[\begin{tikzpicture}[cat,scale=0.8]
\node[h vertex] (in) at (0,0){};
\node[black vertex] (v1) at (0,0){};
\node[vertex] (v2) at (1,-0.8){};
\node[black vertex] (v3) at (1,0.8){};
\node[black vertex] (v4) at (2.5,-0.8){};
\node[vertex] (v5) at (2.5,0.8){};
\node[black vertex] (v6) at (3.5,0){};
\draw[thick] (v1) -- (v2) -- (v3) -- (v1);
\draw[thick] (v4) -- (v5) -- (v6) -- (v4);
\draw[thick] (v2) -- (v4);
\draw[thick] (v3) -- (v5);
\node(v1) at (-0.5,0){$v_1$};
\node(v2) at (0.7,1.1){$v_2$};
\node(v5) at (2.8,1.1){$v_5$};
\node(v3) at (4,0){$v_3$};
\node(v4) at (2.8,-1.1){$v_4$};
\node(v6) at (0.7,-1.1){$v_6$};
\end{tikzpicture}\]
\caption{Open graph $(G,I,O)$ with $I = \{v_1\}$ and $O=\{v_5,v_6\}$ satisfying the uniform equiprobability condition but having not gflow.}
\label{equigf}
\end{figure}
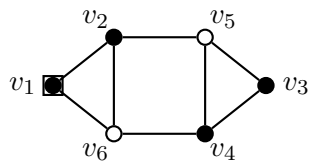

However, in the particular case where $|I|=|O|$, the existence of a gflow implies uniform equiprobability.

\begin{theorem} \label{thmunifgflow}
When $|I|=|O|$, $(G,I,O)$ guarantees uniform equiprobability iff it has a gflow.
\end{theorem}
\begin{proof}
We only  have to prove that uniform equiprobability implies the existence of gflow (the other direction is obvious). We prove the existence of a gflow for $(G,O,I)$ which, according to theorem \ref{thm:revers},  implies the existence of a gflow for $(G,I,O)$. Since $(G,I,O)$ is uniformly equiprobable, the matrix $\smat{A_{G}}{I^C}{O^C}$ is injective, so reversible. Indeed, for any $W\subseteq O^C$, $\smat{A_{G}}{I^C}{O^C}.1_W= \emptyset \iff 1_{Odd(W)\cap I^C} = 0 \implies  Odd(W)\subseteq I\subseteq W\cup I $ so $W=\emptyset$. The matrix $\left({\smat{A_{G}}{I^C}{O^C}}\right)^{-1}$ is the induced matrix of a directed open graph $(H,O,I)$, where $H$ is chosen s.t. vertices in $O$ have no successor. In the following we show that $H$ is a DAG. By contradiction, let $S\subseteq V(H)$ be the shortest cycle in $H$. Notice that  $S\subseteq O^C$ since vertices in $O$ have no successor. 
$\smat{A_{G}}{I^C}{O^C}.(\smat{A_{G}}{I^C}{O^C})^{-1}.1_S = 1_S \iff \smat{A_{G}}{I^C}{O^C}.1_{Odd_H(S)\cap O^C} =1_S \iff Odd_G(Odd_H(S)\cap O^C)\cap I^C = S$. Let $W := Odd_H(S)\cap O^C$. Since $S$ is the shortest cycle, $S\subseteq Odd_H(S)$. Moreover $S\subseteq O^C$ so $S\subseteq W$. Thus $Odd_G(W)\subseteq W\cup I^C$ which implies $W=\emptyset$, so $S=\emptyset$. Thus $H$ is a DAG. 
%
%
%
%
%
%
\end{proof}


Notice that thanks to Theorem \ref{thmunifgflow}  the stepwise condition in the characterisation of gflow can be removed, improving Theorem \ref{thm:gflow}:

\begin{corollary}
When $|I|=|O|$, if  $(G,I,O)$ guarantees uniform strong determinism  iff it has a gflow.
\end{corollary}
\begin{proof}
Uniform strong determinism implies equiprobability which ensures the existence of gflow when $|I|=|O|$.
\end{proof}

\section{Choosing Inputs and Outputs}\label{sec:choosing}

The fact that the characterisation of uniform probability is by internal subsets allows us to have a better view of the following general problem: given a graph, which vertices can be chosen as outputs and inputs for measurement based quantum  information processing.

\begin{definition}
Given a graph $G$, for any $A\subseteq V(G)$,  let $\mathcal{E}_A$ be the collection of  internal sets outside $A$:
$
\mathcal{E}_A := \{S\subseteq V, S\neq \emptyset   \wedge \mbox{Odd}(S) \cap S^C \cap A^C= \emptyset \}
$
\end{definition}
A transversal of a collection $C$ of sets is a set that intersects all the elements of $C$. The set of all transversals of $\mathcal{E}_A$ is $T(\mathcal{E}_A):= \{S^{\prime} \subseteq V, {}^{\forall} S \in \mathcal{E}_A \hspace{0.2cm}S^{\prime} \cap S \not= \emptyset \}$.

\begin{lemma}
If an  open graph $(G,I,O)$ guarantees uniform equiprobability then $O \in T(\mathcal E_{\emptyset})$.
\end{lemma}
\begin{proof}
By contradiction if $W\in \mathcal E_{\emptyset}$ and $W\cap O = \emptyset$, then $Odd(W)\cap W^C =\emptyset$, so $Odd(W)\subseteq W\cup I^C$ which implies $W=\emptyset$. It contradicts the fact that $W\in \mathcal E_{\emptyset}$. 
\end{proof}

%

\begin{theorem}
An open graph $(G,I,O)$ guarantees uniform equiprobability if and only if $O \in T(\mathcal E_{I})$.
\end{theorem}
\begin{proof}
$O\in T(\mathcal E_I)\iff \forall W\in \mathcal E_I, W\cap O\neq \emptyset \iff \forall W\subseteq O^C, W\notin \mathcal E_I \iff \forall W\subseteq O^C, \neg (Odd(W)\cap W^C \cap I^C \wedge W\neq \emptyset) \iff \forall W\subseteq O^C, (Odd(W)\subseteq W\cup I \Rightarrow W=\emptyset)$.
\end{proof}

\begin{theorem}
Given a graph $G$ and two subsets of vertices $I$ and $O$ with $|I|=|O|$, the open graph $(G,I,O)$  guarantees equiprobability iff  $I \in T(\mathcal E_{\emptyset})$ and $O \in T(\mathcal E_{I})$.
\end{theorem}
\begin{proof}
When  $|I|=|O|$, if $(G,I,O)$ guarantees equiprobability then $(G,I,O)$ has a gflow (Theorem \ref{thmunifgflow}) and thus $(G,O,I)$ has a gflow (Theorem \ref{thm:revers}) as well. As a consequence $(G,I,O)$ guarantees uniform equiprobability so  $I \in T(\mathcal E_{\emptyset})$. \end{proof}
{This observation allows a characterisation of  the possible deterministic computations for small graphs.} The main question is,
given a  graph $G$, how to find $I\subseteq V(G)$ and $O\subseteq V(G)$ with $|I|=|O|$ such that $(G,I,O)$ has gflow.

Furthermore it is straightforward to see that :
\begin{lemma}
If an open graph $(G,I,O)$  guarantees uniform equi-probability then $(G,I',O')$ with $I' \subseteq I$ and $O\subseteq O'$ also  guarantees uniform equi-probability.
\end{lemma}
Notice that  gflow and constant probability classes are  also stable by adding new outputs or removing inputs. Thus the interesting problem when choosing inputs and outputs consists of minimizing $|O|$ and  maximizing $|I|$. 

Thus one can take minimal elements in $T(\mathcal E_{\emptyset})$ as inputs  $I$ and then look for minimal elements in $T(\mathcal E_I)$.
If they have the same size then we can conclude that they are a proper input/output pair for deterministic computation. This allows one to characterise the possible deterministic computations for small graphs (as it is not polynomial to compute the big transversal sets). For instance in the case of the $2 \times 3$ grid, the test shows that the minimal number of outputs is 2 and that there are only 3 solutions up to symmetry (see Figure \ref{inout}).
\begin{figure}[!h] 
\[\begin{tikzpicture}[cat,scale=0.8]
\node[h vertex] (in) at (-0.5,0.8){};
\node[black vertex] (v0) at (-0.5,0.8){};
\node[h vertex] (in) at (-0.5,-0.8){};
\node[black vertex] (v1) at (-0.5,-0.8){};
\node[black vertex] (v2) at (1,-0.8){};
\node[black vertex] (v3) at (1,0.8){};
\node[vertex] (v4) at (2.5,-0.8){};
\node[vertex] (v5) at (2.5,0.8){};
\draw[thick] (v0) -- (v1) -- (v2) -- (v4) -- (v5) -- (v3) -- (v0);
\draw[thick] (v3) -- (v2);
\end{tikzpicture}~~~~~~~~~~~~~~~
\begin{tikzpicture}[cat,scale=0.8]
\node[h vertex] (in) at (2.5,0.8){};
\node[black vertex] (v0) at (-0.5,0.8){};
\node[h vertex] (in) at (1,0.8){};
\node[vertex] (v1) at (-0.5,-0.8){};
\node[vertex] (v2) at (1,-0.8){};
\node[black vertex] (v3) at (1,0.8){};
\node[black vertex] (v4) at (2.5,-0.8){};
\node[black vertex] (v5) at (2.5,0.8){};
\draw[thick] (v0) -- (v1) -- (v2) -- (v4) -- (v5) -- (v3) -- (v0);
\draw[thick] (v3) -- (v2);
\end{tikzpicture}~~~~~~~~~~~~~~~
\begin{tikzpicture}[cat,scale=0.8]
\node[h vertex] (in) at (1,0.8){};
\node[vertex] (v0) at (-0.5,0.8){};
\node[h vertex] (in) at (-0.5,-0.8){};
\node[black vertex] (v1) at (-0.5,-0.8){};
\node[vertex] (v2) at (1,-0.8){};
\node[black vertex] (v3) at (1,0.8){};
\node[black vertex] (v4) at (2.5,-0.8){};
\node[black vertex] (v5) at (2.5,0.8){};
\draw[thick] (v0) -- (v1) -- (v2) -- (v4) -- (v5) -- (v3) -- (v0);
\draw[thick] (v3) -- (v2);
\end{tikzpicture}\]
\caption{Uniform deterministic choice of inputs for the $2 \times 3$ grid -- input (resp. output) vertices are represented by squared (resp. white) vertices.}
\label{inout}
\end{figure}
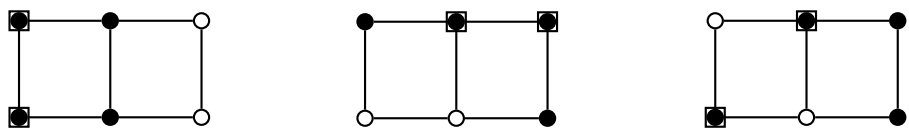

\section{Uniform Constant Probability}
The constant probability case is at the same time the most general case where information is not lost during the measurement and the less understood case.  In this last section, we investigate some properties of such graph states. We show a decomposition theorem into a gflow part and an `even' part, we analyze a transformation, and  we characterise open graphs with constant probability in the special case of one input and one output.
We also prove reversibility in the considered cases. 
\begin{lemma}
An open graph $(G,I,O)$ guarantees uniformly constant probability if and only if it guarantees uniformly constant probability after an $IO-$extension (after which all input/output vertices are of degree 1).
\end{lemma}
\begin{proof}
Let  $(G,I,O)$ be a an open graph and let $(G',I',O)$ be the graph obtained by changing an input $i$ to a noninput and adding a new input vertex $i'$ of degree one connected to $i$.
 $(G,I,O)$ guarantees uniformly constant probability iff it has no violating subsets {\it i.e.} a subset of nonoutput vertices containing inputs and with empty external odd neighborhood  ($Odd(S)\cap S^C=\emptyset$), or a subset of nonoutput vertices whose external odd neighborhood is a proper subset of the inputs.
For the input extension: for any violating set $S$ in $G$ ($Odd(S)\cap S^C$ is a proper subset of the inputs),  $S'=S\cup \{i\}'$ is a violating set for $G'$.  Similarly, for any violating set $S'$ in $G'$, either it does not contain $i'$ and then it is already a violating set in $G$, or $S'\setminus{i'}$ is a violating set in $G$.
For the output extension the violating sets are the same, as in the former graph the output cannot be in a violating set, and in the extended graph  the old output cannot be in a violating set as otherwise the odd neighborhood of the set will contain an output.
\end{proof}
Notice that  this extension also preserves equiprobability and gflow. 

Using the extension we can define graphs for which if there is a violating set then there is a violating set outside the inputs.

\begin{corollary}
Let $(G',I',O')$ be an open graph obtained by two input/output extensions of a given $(G,I,O)$. $(G,I',O')$  satisfies uniform constant probability iff $\nexists W \subseteq  O'^C \cap I'^C\, s.t. \, W\neq \emptyset$,  $Odd(W) \cap W^C \cap I'^C = \emptyset$,  and $Odd(W) \cap W^C \neq \emptyset$ (violating subsets are subsets of $I'^C\cap O'^C$ whose external odd neighborhood is a proper subset of the inputs.)
\end{corollary}
\begin{proof}
We replace each input by  a new input of degree one and a vertex of degree 2 to form a new open graph $(G',I',O')$.
Let $S$ be a violating set that contains a new input $i$ in $G'$, and let $v$ be the vertex of degree 2  neighboring $i$.  $S$ has no odd neighbor outside itself and the inputs so either $S$ contains $v$, in which case $S\setminus \{i\}$ is also a violating set or $S$ does not contain $v$ but contains the old input, thus  $S\setminus \{i\} \cup \{v\}$ is a violating set.
Thus by induction over the number of input vertices, if there is a violating set then there is a violating set outside the inputs.
\end{proof}

Equivalently one can define violating sets to be always containing inputs.

\begin{lemma} \label{lemdec}
If open graph $(G,I,O)$ with $|I|=|O|$ guarantees uniform constant probability then there exists a subgraph  $G'=(V',E')$ of $G$ such that 
$(G',I,O)$ has a  gflow and $V\setminus V'$ has no odd neighborhood.
\end{lemma}

\begin{proof}
 Inductively removing the empty neighborhood subsets ($W$ such that $Odd(W) \cap W^C = \emptyset$) leaves an open graph with gflow.
\end{proof}

\begin{theorem}
An open graph $(G,I,O)$ with $|I|=|O|=1$ such that all input/output vertices are of degree 1 guarantees uniformly constant probability if and only if the subgraph obtained by removing the inputs and outputs is Eulerian (every vertex in $I^C \cap O^C$ is of even degree).
\end{theorem}
\begin{proof}
 We use the extension (note that it preserves the Eulerian property).
If    $I^C\cap O^C$ is Eulerian then there are no violating sets, as a violating set would define an odd cut (odd number of edges going out from the subset) which is impossible as all the vertices have even degree.

For the other direction, by contradiction: consider a  constant probability graph with one input and one output and suppose that there exists a vertex of odd degree and consider the shortest path $P$  between the output and this vertex.  All the vertices of $P $ have the same degree in $P$ and in $G$ (all even except the borders that are odd), thus $Odd(G\setminus P) \cap (G\setminus P)^C= \emptyset$ which contradicts the hypothesis.
\end{proof}

\section{Open Questions}
This work raises several open questions, from the structural point of view.  For example,  it is not known  whether the uniform constant probability case is reversible when $|I|=|O|$.  From a complexity  perspective: is it possible to derive a polynomial algorithm to characterise the uniform equiprobability class and the uniform constant probability class? 
Is it possible to derive an efficient algorithm for finding inputs and ouputs?

\section*{Acknowledgements} The authors want to  thank E. Kashefi for discussions.  This work is supported by CNRS-JST Strategic French-Japanese Cooperative Program, and Special Coordination Funds for Promoting Science and Technology in Japan.

\bibliographystyle{plain}

\end{document}